\documentclass[aps,nofootinbib,superscriptaddress]{revtex4}

\usepackage{amsmath}
\usepackage{amsfonts}
\usepackage{graphicx}

\usepackage[T1]{fontenc}
\usepackage[latin9]{inputenc}
\usepackage{color}
\usepackage{prettyref}
\usepackage{float}
\usepackage{amsthm}
\usepackage{amsbsy}
\usepackage{amstext}
\usepackage{amssymb}
\usepackage{graphicx}
\usepackage{bbm}

\makeatletter

\floatstyle{ruled}
\newfloat{algorithm}{tbp}{loa}
\providecommand{\algorithmname}{Algorithm Class}
\floatname{algorithm}{\protect\algorithmname}

\theoremstyle{plain}
\newtheorem{thm}{\protect\theoremname}
  \theoremstyle{plain}
  \newtheorem{cor}[thm]{\protect\corollaryname}

\makeatother

\providecommand{\corollaryname}{Corollary}
\providecommand{\theoremname}{Theorem}

\usepackage{natbib}
\usepackage{hyperref}

\begin{document}
\global\long\def\ket#1{\left| #1 \right\rangle }

\global\long\def\bra#1{\left\langle #1 \right|}

\global\long\def\braket#1#2{\left\langle #1 | #2 \right\rangle }

\global\long\def\ketbra#1#2{|#1\rangle\!\langle#2|}

\global\long\def\braopket#1#2#3{\bra{#1}#2\ket{#3}}

\global\long\def\Tr{\text{Tr}}

\global\long\def\Pr#1{\text{Pr}(#1) }




\title{Negative Quasi-Probability as a Resource for Quantum Computation}

\author{Victor Veitch}

\author{Christopher Ferrie}
\affiliation{Institute for Quantum Computing and Department of Applied Mathematics,
University of Waterloo, Waterloo, Ontario, Canada, N2L 3G1}

\author{David \surname{Gross}}
\affiliation{Institute for Physics, University of Freiburg, Rheinstrasse 10, D-79104 Freiburg, Germany}

\author{Joseph Emerson}
\affiliation{Institute for Quantum Computing and Department of Applied Mathematics,
University of Waterloo, Waterloo, Ontario, Canada, N2L 3G1}

\begin{abstract}
A central problem in quantum information is to determine the minimal
physical resources that are required for quantum computational speedup
and, in particular, for fault-tolerant quantum computation. We establish
a remarkable connection between the potential for quantum speed-up
and the onset of negative values in a distinguished quasi-probability
representation, a discrete analog of the Wigner function for quantum
systems of odd dimension. This connection allows us to resolve an
open question on the existence of bound states for magic-state distillation:
we prove that there exist mixed states outside the convex hull of
stabilizer states that cannot be distilled to non-stabilizer target
states using stabilizer operations. We also provide an efficient simulation
protocol for Clifford circuits that extends to a large class of mixed
states, including bound universal states.
\end{abstract}
\maketitle
While it is widely believed that quantum computers can solve certain
problems with exponentially fewer resources than their classical counterparts,
the scope of the physical resources of the underlying quantum systems
that enable universal quantum computation is not well understood.
For example, for the standard circuit model of quantum computation,
Vidal has shown that high-entanglement is necessary for an exponential
speed-up \cite{Vidal2003Efficient}; however, it is also known that access to high-entanglement
is not sufficient \cite{Gottesman1997Stabilizer}. Moreover, in alternative
models of quantum computation such as DQC1 \cite{Knill1998Power}, algorithms
that may be performed on highly-mixed input states appear to be more
powerful than classical computation even though there appears to be
a negligible amount of entanglement in the underlying quantum system
\cite{Datta2005Entanglement}. This suggests that large amounts of entanglement, purity
or even coherence may not be necessary resources for quantum-computational
speed-up. One of the central open problems of quantum information
is to understand which sets of quantum resources are jointly necessary
and sufficient to enable an exponential speed-up over classical computation.
Any solution to this important problem may point to more practical
experimental means of achieving the benefits of quantum computation.

The question of whether a restricted subset of quantum theory is still sufficient for a given task 
is meaningful when there is a specific context that divides the full set of possible quantum operations
into two classes: the restricted subset of operations that are accessible or easy to implement and the remainder that
are not. In such a context it is then natural to consider the difficult operations
as resources and ask how much, if any, of these resources are required. For example, a common paradigm in quantum communication
is that of two or more spatially separated parties for which local quantum operations and classical communication
define a restricted set of operations that are accesible or ``free resources'', whereas joint quantum operations are not free; in this context entanglement is the natural
resource for quantum communication. Here we are interested not in
quantum communication but in the power of quantum computation, and in particular
the practically relevant case of fault-tolerant quantum computation.  Transversal unitary gates,
i.e., gates that do not spread errors within each code block, play
a critical role  in fault-tolerant quantum computation.
Recent theoretical work has shown that  a set of quantum
gates which is both universal and transversal, and hence fault-tolerant
\cite{Eastin2009Restrictions,Zeng2007Transversality,Chen2008Subsystem}, does not exist. That is, any scheme for fault
tolerant quantum computation divides quantum operations into two classes:
those with a fault-tolerant implementation -- these are the ``free resources'' --  and the remainder -- these are not free but are required to achive universality.  For a fixed fault tolerant scheme the critical question is:
what are necessary and sufficient physical resources to promote
fault-tolerant computation to universal quantum computation?

Most of the best known fault tolerant schemes are built around the
well-known stabilizer formalism \cite{Gottesman2006Quantum}, in which a distinguished
set of preparations, measurements, and unitary transformations (the
{}``stabilizer operations'') have a fault tolerant implementation.
Stabilizer operations also arise naturally in some physical systems with topological
order \cite{Moore1991Nonabelions, Lloyd2002Quantum, Doucot2002Pairing}. As described above,
the transversal set of stabilizer operations do not give a universal
gate set and must be supplemented with some additional (non-stabilizer)
resource.  A celebrated scheme for overcoming this limitation is the
magic state model of quantum computation devised by Bravyi and Kitaev
\cite{Bravyi2004Universal} where the additional resource is a set of ancilla
systems prepared in a some (generally noisy) non-stabilizer quantum
state. Hence, in this important paradigm, the question of which physical resources are
required for universal fault-tolerant quantum computation reduces to the following: which non-stabilizer
states are necessary and sufficient to promote stabilizer computation
to universal quantum computation? 

In this paper we identify a non-trivial closed, convex subset of the space of
quantum states which we prove is incapable of producing universal fault-tolerant quantum computation. In particular,
we show that this convex subset strictly contains the convex hull
of stabilizer states, and thereby prove that there exists a class
of \emph{bound universal states}, i.e. states that can not be prepared
from convex combinations of stabilizer states and yet are not useful
for quantum computation. Thus our proof of the existence of  bound universal
states resolves in the negative the open problem raised
by Bravyi and Kitaev  \cite{Bravyi2004Universal}
of whether \emph{all} non-stabilizer states promote stabilizer computation
to universal quantum computation. Furthermore, we devise an efficient simulation algorithm
for the subset of quantum theory that consists of operations from
the stabilizer formalism acting on inputs from our non-universal region,
which includes mixed states both inside and outside the convex hull
of stabilizer states. This simulation scheme is an extension of the
celebrated Gottesman-Knill theorem \cite{Gottesman1997Stabilizer,Aaronson2004Improved} to a broader class of input state and should be of independent interest. 

Our theoretical method for proving these results is to construct a classical, local
hidden variable model for the subtheory of quantum theory that consists of the stabilizer formalism and then determine the scope of additional quantum resources that are also described by this model. 
Indeed our local hidden variable model is a distinguished quasi-probability representation with non-negative elements. 
For a $d$ dimensional quantum system there are many possible ways to represent
arbitrary quantum states as quasi-probability distributions over a
phase space of $d^{2}$ points and projective measurements as conditional
quasi-probability distributions over the same space (see \cite{Ferrie2009Framed,Ferrie2010Necessity} for further details). Perhaps unsurprisingly, it has been shown that the full quantum theory
can not be represented with non-negative elements in any such representation  \cite{Ferrie2008Frame,Spekkens08,Ferrie2009Framed,Ferrie2010Necessity}. However, one might expect
that a subtheory of quantum theory that is inadequate for quantum speed-up might be represented non-negatively, i.e. as a true
probability theory, in some natural choice of quasi-probability representation. For the context described above, we seek a quasi-probability representation reflecting our natural operational restriction, in particular, we require that stabilizer states and projective
measurements onto stabilizer states have non-negative representation
and that unitary stabilizer operations (i.e., Clifford transformations) correspond to stochastic processes.
Conveniently, for
quantum systems with odd Hilbert space dimension such a representation
is already known to exist: this is the discrete Wigner function picked
out by Gross \cite{Gross2006Hudsons,Gross2007Nonnegative} from the broad class defined
by Gibbons \emph{et al} \cite{Gibbons2004Discrete}.  In such a representation it is natural to examine
whether the resouces that are necessary or sufficient for quantum speed-up correspond to those that are not represented by non-negative elements of the representation. 

With this insight in hand the results of this paper may now be stated
more carefully:

\textbf{Classically efficient simulation of positive Wigner functions:} The set of fault tolerant quantum logic gates in the stabilizer formalism
are known as the \emph{Clifford} gates. Our first contribution is
an explicit simulation protocol for quantum circuits composed of Clifford
gates acting on input states with positive discrete Wigner representation.
We also allow arbitrary product measurements with positive discrete
Wigner representation. This simulation is efficient (linear) in the
number of input registers to the quantum circuit. This simulation
scheme is an extension of the celebrated Gottesman-Knill theorem
and should be of independent interest.

\textbf{Negativity is necessary for magic state distillation:} This simulation protocol implies that states outside the stabilizer
formalism with positive discrete Wigner function (bound universal
states) are not useful for magic state distillation. The second contribution
of this paper is to give a direct proof of this fact exploiting only
the observation that negative discrete Wigner representation can not
be created by stabilizer operations. This proof has a more general range of applicability than the efficient simulation scheme and also makes clear
the conceptual importance of negative quasi-probability as a \emph{resource}
for stabilizer computation.

\textbf{Geometry of positive Wigner functions:} The set of quantum states with positive discrete Wigner function strictly
contains the set of (convex combinations of) stabilizer states. To
prove this fact we determine the geometry of the region of quantum
state space with positive discrete Wigner representation. Concretely,
we show that the facets of the classical probability simplex defining
the discrete Wigner function are also facets of the polytope with
the (pure) stabilizer states as its vertices. Since there are many
more facets of the stabilizer polytope than of the simplex this suffices
to show the existence of non-stabilizer states with positive representation. 

The paper concludes with a discussion of this work and some avenues
for future exploration.

\subsection*{Previous Work}

The Gottesman-Knill theorem provides an efficient classical simulation
protocol for circuits of Clifford unitaries acting on stabilizer states.
This result deals with pure qubit stabilizer state inputs and simulates
the evolution of the full quantum state. The simulation scheme of
the present paper deals with odd dimensional systems, makes no distinction
between mixed state and pure state input, and allows the simulation
of a large class of non-stabilizer states. However, our scheme constructs
a classical circuit with the same outcome probabilities as the quantum
circuit and does not recover the evolution of the full quantum state.

A number of papers have addressed the question of which ancilla states
enable universal quantum computation for the magic state model in
qubit systems \cite{Campbell_bound_magic_states, Campbell_magic_state_catalysis,campbell_struct_of_dist_protocols,Reichardt2005Quantum,Reichardt2009ErrorDetectionBased,Reichardt2009Quantum,ratanje_gen_entanglement,vanDam2010Noise_for_higher_dim_sys}.
The most directly comparable result is the demonstration by Campbell
and Browne \cite{Campbell_bound_magic_states} that for any protocol
on the input $\rho^{\otimes n}$ there exists a $\rho$ outside the
convex hull of stabilizer states that maps to a convex combination
of stabilizers. As $n$ grows these states are known to exist only
within some arbitrarily small distance $\epsilon$ of the convex hull
of stabilizer states. By contrast, the present result implies the
existence of states a fixed distance from the hull which are not distillable
by any protocol.

The present result is complementary to previous work connecting negativity
in discrete Wigner function type representations to quantum computational
speedup \cite{Cormick2006Classicality,Galvao2005Discrete,vanDam2010Noise_for_higher_dim_sys}.
In particular, van Dam and Howard \cite{vanDam_stab_noise_threshold}
have used techniques of this type to derive a bound on the amount
of depolarizing noise a state can withstand before entering the stabilizer
polytope. Their work deals only with prime dimensional systems, and
in this case it turns out that the noise threshold they derive is
the same as the amount of noise required for their {}``maximally
robust'' state to enter the region of positive states.

\section{Background}

\subsection{The Stabilizer Formalism}

Known schemes for fault tolerant quantum computation allow for only
a limited set of operations to be implemented directly on the encoded
quantum information. For most known fault tolerance schemes this restricted
set is the stabilizer operations consisting of preparation and measurement
in the computational basis and a restricted set of unitary operations.
This restricted set is the Clifford group, and we now review the important
parts of its structure for qudit systems \cite{Gross2007Evenly}.
The primitive object about which the the stabilizer formalism is built
is the Heisenberg-Weyl group, an extension of qubit Pauli group to
odd dimensional systems. We will define the Heisenberg-Weyl group
in terms of generalized $X$ and $Z$ operators. In odd \emph{prime}
dimension $d$ these are given by their action on computational basis
states, 
\begin{eqnarray*}
X\ket x & = & \ket{x+1\text{ mod }d}\\
Z\ket x & = & \omega^{x}\ket x,
\end{eqnarray*}
 where $\omega=\exp\left(\frac{2\pi i}{d}\right)$ is a $d$th primitive
root of unity. The Heisenberg-Weyl operators are the $d^{2}$ operators
generated by $X$ and $Z$, which have a group structure if we include
phases. A general element of the group is given as 
\[
T_{(a_{1},a_{2})}=\omega^{-\frac{a_{1}a_{2}}{2}}Z^{a_{1}}X^{a_{2}}\ (a_{1},a_{2})\in\mathbb{Z}_{d}\times\mathbb{Z}_{d}.
\]
$\mathbb{Z}_{d}$ is the finite field of $d$ elements, and the field
$\mathbb{Z}_{d}\times\mathbb{Z}_{d}$ will form the {}``phase space''
underpinning the discrete Wigner function defined below.

This definition applies only for prime dimension, but it is easily
promoted to arbitrary odd dimension. In this case the Heisenberg-Weyl
operators are defined to be tensor products of the Heisenberg-Weyl
operators of the factor spaces. For a system with composite Hilbert
space $H_{\boldsymbol{a}}\otimes H_{\boldsymbol{b}}\otimes\dots\otimes H_{\boldsymbol{u}}$
the Heisenberg-Weyl operators may be written as: 
\begin{align*}
T_{(a_{1},a_{2})\oplus(b_{1},b_{2})\dots\oplus(u_{1},u_{2})} & \equiv T_{(a_{1},a_{2})}\otimes T_{(b_{1},b_{2})}\dots\otimes T_{(u_{1},u_{2})}.
\end{align*}
 If the dimension of $H_{j}$ is $d_{j}$ the vector $(j_{1},j_{2})$
is an element of $\mathbb{Z}_{d_{j}}\times\mathbb{Z}_{d_{j}}$ and
the vector $(a_{1},a_{2})\oplus(b_{1},b_{2})\dots\oplus(u_{1},u_{2})$
is an element of $\left(\mathbb{Z}_{d_{a}}\times\mathbb{Z}_{d_{a}}\right)\times\left(\mathbb{Z}_{d_{b}}\times\mathbb{Z}_{d_{b}}\right)\dots\times\left(\mathbb{Z}_{d_{u}}\times\mathbb{Z}_{d_{u}}\right)$.
We denote the Heisenberg-Weyl group for dimension $d$ as $\mathcal{W}_{d}$.

The Clifford operators, $\mathcal{C}_{d}$, are the set of unitary
operators that map Heisenberg-Weyl operators to Heisenberg-Weyl operators
under conjugation: 
\[
U\in\mathcal{C}_{d}\iff UT_{\boldsymbol{u}}U^{\dagger}\in\mathcal{W}_{d}\ \forall T_{\boldsymbol{u}}\in\mathcal{W}_{d}.
\]
 This is the normalizer of the Heisenberg-Weyl group. This group has
many interesting features. For instance, operations of the Clifford
group on computational basis input states can be efficiently simulated
even though they create large amounts of entanglement.

For this paper it will be necessary to understand the Clifford group
in terms of its representation over finite fields. Because it enormously
simplifies the presentation we will restrict ourselves to working
in the case where the system is composed of $n$ subsystems with a
common Hilbert space dimension $p$ a prime. However, many of the
important results carry over for arbitrary odd dimensional systems \cite{Gross2006Hudsons}.
The main result that we want is that a Clifford operation $U_{\boldsymbol{F},\boldsymbol{a}}\in\mathcal{C}_{p^{n}}$
can be specified as, 
\[
U_{\boldsymbol{F},\boldsymbol{a}}T_{\boldsymbol{u}}U_{\boldsymbol{F},\boldsymbol{a}}^{\dagger}=T_{\boldsymbol{Fu}+\boldsymbol{a}}.
\]
 That is, any Clifford unitary is uniquely specified by its action
on the Pauli group and this is given by induced action on the label
$\boldsymbol{u}\in\left(\mathbb{Z}_{p}\times\mathbb{Z}_{p}\right)^{n}$.
The matrix $\boldsymbol{F}$ is a $2n\times2n$ symplectic matrix
with entries in $\mathbb{Z}_{p}$ and $\boldsymbol{a}\in\left(\mathbb{Z}_{p}\times\mathbb{Z}_{p}\right)^{n}$.
Clifford operations generally factor as 
\[
U_{\boldsymbol{F},\boldsymbol{a}}=U_{\boldsymbol{F}}T_{\boldsymbol{a}},
\]
where $\boldsymbol{F}\in\text{Sp}(2n,\mathbb{Z}_{p})$, the group
of symplectic matrices of size $2n$ with entries in $\mathbb{Z}_{p}$.
A matrix is symplectic if it preserves the symplectic product, which
may be defined in a natural way on finite fields. This structure is
not important for the present paper so we do not cover it in detail.
What is important is that this representation of the Clifford group
has a size linear in $n$ and thus can be easily tracked by a classical
computer; this is at the heart of simulation results about the Clifford
group.

Finally, we define stabilizer states to be any state that can be prepared
by applying a Clifford unitary to a computational basis state. These
states are important because they are the only pure states that can
be prepared using our restricted operation set.

\subsection{Magic State Distillation}

It is possible to implement stabilizer operations fault tolerantly,
but these operations do not suffice for universal quantum computation.
To promote stabilizer computation to universal quantum computation
some additional resource is required. This additional resource will
be subject to large amounts of noise, so the question becomes: which
non-stabilizer resources can be used to promote stabilizer computation
and how can this be done? The first of these questions is the subject
of this paper. The second question finds a particularly elegant solution
in the form of magic state distillation \cite{Bravyi2004Universal,Anwar2012Qutrit,Campbell2012Magic}.

Magic state distillation protocols aim to consume a large number of
copies of a non-stabilizer qudit input state $\rho_{\text{in}}$ to
produce a single non-stabilizer qudit output state $\rho_{\text{out}}$
with higher fidelity to some non-stabilizer pure state. This output
state is then consumed to implement some non-Clifford unitary gate.
These protocols have the following structure: 
\begin{itemize}
\item Prepare a number of copies of the input state $\rho_{\text{in}}^{\otimes n}$. 
\item Perform some Clifford operation on $\rho_{\text{in}}^{\otimes n}$. 
\item Measure Heisenberg-Weyl observables on the last $n-1$ registers and
post select on the outcome. 
\end{itemize}
When these protocols succeed the first register will be the output
state $\rho_{\text{out}}$. Typically these protocols work iteratively,
repeatedly consuming $\rho_{\text{in}}^{\otimes n}$ until $n$ copies
of $\rho_{\text{out}}$ have been produced and then using $\rho_{\text{out}}^{\otimes n}$
as the input to the protocol to produce $\rho_{\text{out}}^{'}$ and
so on. The protocols we deal with here encompass but do not require
this iterative structure.

It is not clear if for all input states $\rho_{\text{in}}$ there
exists a protocol to produce an output state with arbitrarily high
fidelity to some non-stabilizer input state. We call states for which
such a protocol exists {}``distillable'', and one of our results
is to show that not all non-stabilizer states are distillable.

\subsection{Quasi-Probability Representations}

There is a long history of studying negativity in quasi-probability
representations.  The most notable examples come from quantum optics where the Wigner function \cite{Kenfack2004Negativity} and the $Q$ and $P$ functions \cite{Mandel1986NonClassical} play prominent roles.
However, such approaches typically suffer in significance due to the
problem of non-uniqueness of the choice of representation.  While a
quantum state may correspond to a negative-valued quasi-probability
function in one choice of quasi-probability representation, in another
choice that same state can be positive, and hence a valid classical
probability density. In Reference \cite{Ferrie2010Necessity},
two of us proved that for any choice of quasi-probability representation
in which both quantum states and measurements are represented, at
least some of the states and measurements must take on negative-values.
However, this result still leaves open the possibility that certain
subsets of quantum states and measurements may be represented positively leading to a classical
probability model for the corresponding subset of quantum operations \cite{Wallman2012Nonnegative}.
When the restricted subtheory is prescribed by an operational restriction, this question takes on a precise and relevant meaning.  Indeed, such an approach has been considered already by Schack and
Caves, who constructed classical probability models for few-qubit
NMR experiments \cite{Schack1999Classical} and thereby stimulated an important
discussion of what kinds of resources might be required for universal
quantum computation.

Our approach is to exploit the freedom in choice of quasi-probability representations \cite{Ferrie2011Quasiprobability} in order to
align the positive subtheory with the operational restriction defining
the error-free resources in the magic-state model. We seek a quasi-probability
representation for which stabilizer states and projective measurements
onto stabilizer states have positive representation and for which
stabilizer transformations correspond to stochastic processes. This
first step is easy given that such a representation already exists;
this distinguished representation is the discrete Wigner function picked out by Gross \cite{Gross2006Hudsons,Gross2007Nonnegative}
from the broad class defined by Gibbons \emph{et al} \cite{Gibbons2004Discrete}.

\subsection{The discrete Wigner Function}

Our approach is to represent Clifford operations as stochastic processes
over a discrete phase space. Intuitively, if the dynamics of the quantum
system admit a representation as a classical statistical process then
it should not be sufficient for universal quantum computation. To
that end, we look for a quasi-probability representation for quantum
theory where stabilizer resources are represented positively. This
is the discrete Wigner function.

The discrete Wigner representation of a state $\rho\in L(\mathbb{C}^{p^{n}})$
is a quasi-probability distribution over $\mathbb{Z}_{d}\times\mathbb{Z}_{d}$,
which can be thought of as $d$ by $d$ grid. This grid is the discrete
analogue of the phase space of classical mechanics. The map taking
quantum states to quasi-probability distributions on discrete phase
space is uniquely specified by a set of \emph{phase space point operators}
$\{A_{\boldsymbol{u}}\}$ (defined below). For each point $\boldsymbol{u}$
in the discrete phase space there is a corresponding operator $A_{\boldsymbol{u}}$
and the value of discrete Wigner representation of $\rho$ at this
point is given as, 
\[
W_{\rho}(\boldsymbol{u})=\frac{1}{d}\Tr(A_{\boldsymbol{u}}\rho).
\]
A quantum measurement with POVM $\{E_{k}\}$ is represented by assigning
conditional (quasi-)probability functions over the phase space to
each measurement outcome, 
\[
W_{E_{k}}(\boldsymbol{u})=\Tr(A_{\boldsymbol{u}}E_{k}).
\]
In the case where $W_{E_{k}}(\boldsymbol{u})\ge0\;\forall\boldsymbol{u}$,
this can be interpreted classically as the probability of getting
outcome $k$ given that the system is actually at point $\boldsymbol{u}$,
$W_{E_{k}}(\boldsymbol{u})=\text{Pr}(\text{outcome }k|\text{location }\boldsymbol{u})$.
If both $W_{\rho}(\boldsymbol{u})$ and $W_{E_{k}}(\boldsymbol{u})$
are positive then the law of total probability gives the probability
of getting outcome $k$ from a measurement of state $\rho$, 
\[
\text{Pr}(k)=\sum_{\boldsymbol{u}}W_{\rho}(\boldsymbol{u})W_{E_{k}}(\boldsymbol{u}).
\]
In fact this prediction reproduces the Born rule even when $W_{\rho}(\boldsymbol{u})$
or $W_{E_{k}}(\boldsymbol{u})$ take on negative values.

We say a state $\rho$ has positive representation if $W_{\rho}(\boldsymbol{u})\ge0\ \forall\boldsymbol{u}\in\mathbb{Z}_{d}^{n}\times\mathbb{Z}_{d}^{n}$
and negative representation otherwise. We will say a measurement with
POVM $M=\{E_{k}\}$ has positive representation if $W_{E_{k}}(\boldsymbol{u})\ge0\ \forall\boldsymbol{u}\in\mathbb{Z}_{d}^{n}\times\mathbb{Z}_{d}^{n},\ \forall E_{k}\in M$
and negative representation otherwise.

The phase space point operators are defined in terms of the Heisenberg-Weyl
operators as, 
\begin{eqnarray*}
A_{\boldsymbol{0}} & = & \sum_{\boldsymbol{u}}T_{\boldsymbol{u}},\ A_{\boldsymbol{u}}=T_{\boldsymbol{u}}A_{\boldsymbol{0}}T_{\boldsymbol{u}}^{\dagger}.
\end{eqnarray*}
 These operators are Hermitian so the discrete Wigner representation
is real-valued. There are $d^{2}$ such operators for $d$-dimensional
Hilbert space; they are informationally complete and orthogonal in
the sense that $\Tr(A_{\boldsymbol{u}}A_{\boldsymbol{v}})=d\delta(\boldsymbol{u},\boldsymbol{v}).$

These operators have several important features reflecting the salient
properties of the discrete Wigner representation\cite{Gross2006Hudsons,Gibbons2004Discrete}: 
\begin{enumerate}
\item (Discrete Hudson's theorem) if $\ket S$ is a stabilizer state then
$\Tr(A_{\boldsymbol{u}}\ketbra SS)\ge0\ \forall\boldsymbol{u}$, and
the stabilizer states are the only pure states satisfying this property.
That is, a pure state has positive representation if and only if it
is a stabilizer state. 
\item Clifford operators have the action $U_{\boldsymbol{F},\boldsymbol{a}}A_{\boldsymbol{u}}U_{\boldsymbol{F},\boldsymbol{a}}^{\dagger}=A_{\boldsymbol{Fu}+\boldsymbol{a}}.$
This means that 
\[
W_{U_{\boldsymbol{F},\boldsymbol{a}}\rho U_{\boldsymbol{F},\boldsymbol{a}}^{\dagger}}(v)=W_{\rho}(\boldsymbol{F}^{-1}\left(\boldsymbol{v}-\boldsymbol{a}\right))
\]
 so that Clifford transformations map to permutations of the underlying
phase space and, in particular, Clifford operations preserve positive
representation. 
\item For $\rho=\sum_{\boldsymbol{u}}p_{\boldsymbol{u}}A_{\boldsymbol{u}}$
and $\sigma=\sum_{\boldsymbol{u}}q_{\boldsymbol{u}}A_{\boldsymbol{u}}$
the trace inner product is $\Tr(\rho\sigma)=d\sum_{\boldsymbol{u}}p_{\boldsymbol{u}}q_{\boldsymbol{u}}$; 
\item The phase point operations in dimension $d^{n}$ are tensor products
of $n$ copies of the $d$ dimension phase space point operators. 
\end{enumerate}

\section{Negative discrete Wigner Representation is Necessary for Computational
Speedup}

We now establish that any quantum computation consisting of stabilizer
operations acting on product input states with positive representation
can not produce an exponential computational speed-up. To this end
we give an explicit efficient classical simulation protocol for such
circuits. Like the Gottesman-Knill protocol our scheme allows for
the simulation of pure state stabilizer inputs to circuits composed
of Clifford transformations and stabilizer measurements. However,
our simulation scheme extends the Gottesman-Knill result in several
ways. First, it applies to systems of qudits rather than qubits. Second,
it applies to mixed state inputs. Thirdly, and most remarkably, it
applies to some non-stabilizer resources - namely those with positive
discrete Wigner representation.

Any particular run of a quantum algorithm on $n$ registers will produce
a string $\boldsymbol{k}$ of $n$ measurement outcomes. These outcomes
occur at random and we assign the random variable $K_{\text{quant}}$
to be the algorithm output. The algorithm can then be considered as
a way of sampling outcomes according to the distribution $\Pr{K_{\text{quant}}=k}$.
To simulate a quantum algorithm it suffices to give a simulating algorithm
which samples from the distribution $\text{Pr}(K_{\text{quant}}=\boldsymbol{k}$),
which is what we do here. Notice that this form of simulation does
not allow us to actually infer the distribution of outcomes, but it
does suffice for many important tasks (for example, estimating the expected
outcome). 

The type of algorithms we treat here take the following form (see
Figure \ref{fig:Example-Circuit-1} for an example): 

\begin{algorithm}[H]
\begin{enumerate}
\item Prepare an initial $n$ qudit input state $\rho_{1}\otimes\dots\otimes\rho_{n}\in\rho\in L(\mathbb{C}^{p^{n}})$
where $\rho_{1},\dots,\rho_{n}$ have positive discrete Wigner representation. 
\item Until all registers have been measured:

\begin{enumerate}
\item Apply a Clifford unitary gate $U_{\boldsymbol{F}}$, labeled by the
symplectic transformation $\boldsymbol{F}\in\text{Sp}(2n,d)$. 
\item Measure the final qudit register using a measurement with positive
discrete Wigner representation. Record the outcome $k_{j}$ of measurement
the $j$th register. Further steps in the computation may be conditioned
on the outcome of this measurement. 
\end{enumerate}
\end{enumerate}
\caption{Family of Simulable Quantum Algorithms\label{alg:Simulable-Quantum-Circuit}\protect \\
Algorithms in this class sample strings of measurement outcomes $\boldsymbol{k}$
according to the distribution $\Pr{K_{\text{quant}}=\boldsymbol{k}}$
determined by the Born rule.}
\end{algorithm}

\begin{figure}
\includegraphics[width=7cm]{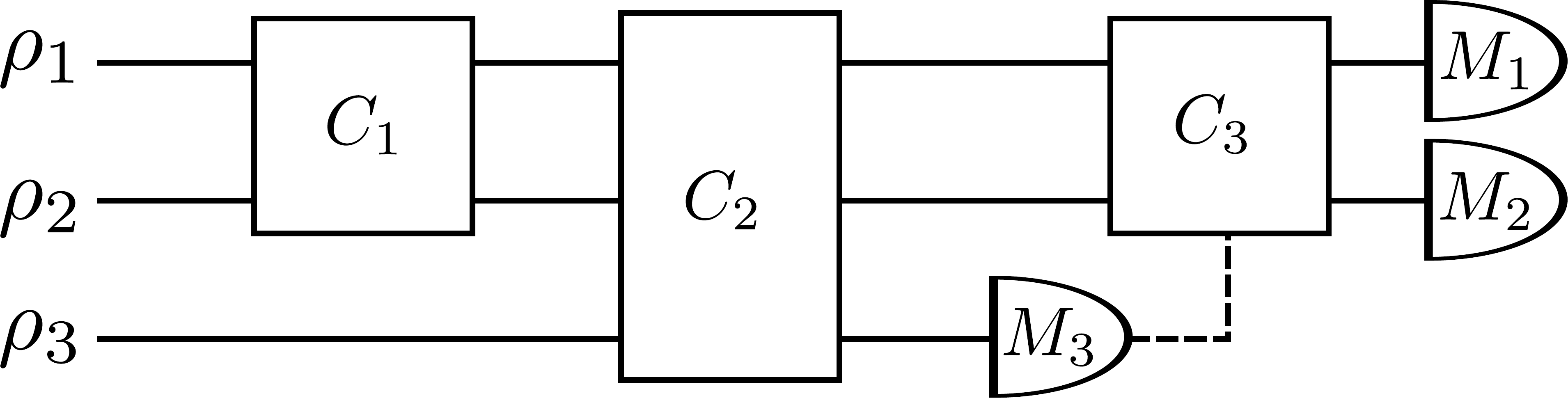}

\caption{Example of a stabilizer circuit:
$\rho_{i}$ have positive representation, $C_{i}$ are Clifford gates
and $M_{i}$ have positive representation.
The choice of gate $C_{3}$ can be conditioned on the outcome of measurement
$M_{3}$. \label{fig:Example-Circuit-1}}
\end{figure}

Notice that there is no loss of generality in considering only symplectic
Clifford transformations as the Heisenberg-Weyl component can be rolled
into the measurement.

The essential idea for the simulation is to take seriously the hidden
variable model the restrictions allow us. In the discrete Wigner picture
the system begins at point $\boldsymbol{u}$ in the discrete phase
space, which is unknown but definite and fixed. The effect of $U_{\boldsymbol{F}}$
is to move the system from the point $\boldsymbol{u}$ to the point
$\boldsymbol{F}\boldsymbol{u}$, and measurement amounts to checking
some region of the phase space to see if it contains the system. Since
the vector $\boldsymbol{u}$ and matrix $\boldsymbol{F}$ are size
$2n$ with entries from $\mathbb{Z}_{d}$ it is computationally efficient
to classically store and update the system's location. Of course,
a (positively represented) quantum state corresponds to a probability
density over the space so we must treat this a little more carefully.
The simulation protocol is: 

\begin{algorithm}[H]
\begin{enumerate}
\item Sample $\boldsymbol{u}\in\mathbb{F}_{d}^{2n}$ according to the distribution
$W_{\rho_{1}\otimes\dots\otimes\rho_{n}}(\boldsymbol{u})=W_{\rho_{1}}(\boldsymbol{u}_{1})W_{\rho_{2}}(\boldsymbol{u}_{2})\dots W_{\rho_{n}}(\boldsymbol{u}_{n}).$ 
\item Repeat until all registers have been measured:

\begin{enumerate}
\item If the unitary $U_{\boldsymbol{F}}$ is applied then update $\boldsymbol{u}\rightarrow\boldsymbol{Fu}$. 
\item If the measurement $M$ with corresponding POVM $\{E_{k}\}$ is made
on the last register of the quantum circuit then report outcome $k$
with probability $W_{E_{k}}(\boldsymbol{u}_{m})$ where $\boldsymbol{u}_{m}$
is the ontic position of the last qudit system, defined by $\boldsymbol{u}=\boldsymbol{u}_{1}\oplus\boldsymbol{u}_{2}\dots\oplus\boldsymbol{u}_{m}.$
If the quantum algorithm conditions further steps on the outcome of
measurement on this register then condition further steps of the simulation
on measurement outcome $k$.
\end{enumerate}
\end{enumerate}
\caption{Classical Simulation Algorithm\label{alg:Equivalent-Classical-Circuit}\protect \\
Algorithms in this class sample strings of measurement outcomes $\boldsymbol{k}$
according to the distribution $\Pr{K_{\text{class}}=\boldsymbol{k}}$}
\end{algorithm}

Our claim is that the classical algorithm in Algorithm Class \ref{alg:Equivalent-Classical-Circuit}
efficiently simulates the corresponding quantum algorithm in Algorithm Class \ref{alg:Simulable-Quantum-Circuit}.
More precisely,
\begin{thm}
An $n$ qudit quantum algorithm belonging to Algorithm Class \prettyref{alg:Simulable-Quantum-Circuit}
is simulable by the corresponding $2n$ dit classical algorithm in Algorithm Class \prettyref{alg:Equivalent-Classical-Circuit} in the sense that the
distribution of outcomes $\boldsymbol{k}$ is the same for both algorithms,
$\Pr{K_{\text{class}}=\boldsymbol{k}}=\Pr{K_{\text{quant}}=\boldsymbol{k}}$.\end{thm}
\begin{proof}
The input to the classical circuit is a $2n$ dit string and the transformations
are all matrices of size $2n$ with entries in $\mathbb{Z}_{d}$ so
the $2n$ dit portion of the claim is obvious.

To show that this protocol genuinely simulates the circuit it suffices
to show any string of measurement outcomes $\boldsymbol{k}=\left(k_{1}k_{2}\dots k_{n}\right)$
occurs with the same probability for both the original circuit and
the simulation. Lets first consider probability distribution $\Pr{k_{n}}$
of the outcomes of the first measurement. In the quantum circuit the
preparation $\rho_{1}\otimes\dots\otimes\rho_{n}$ is passed to the
(possibly identity) gate $U_{\boldsymbol{F}}$ and measurement $M_{n}$
with corresponding POVM $\{E_{k_{n}}\}$ is applied to the $n$th
register. The probability of getting outcome $k_{n}$ is then: 
\begin{eqnarray}
\text{Pr}_{\text{quant}}(k_{n}) & = & \Tr(U_{\boldsymbol{F}}\rho_{1}\otimes\dots\otimes\rho_{n}U_{\boldsymbol{F}}^{\dagger}\mathbb{I}\otimes\dots\otimes\mathbb{I}\otimes E_{k_{n}})\nonumber \\
 & = & \sum_{\boldsymbol{v}\in\mathbb{Z}_{d}^{2n}}W_{U_{\boldsymbol{F}}\rho_{1}\otimes\dots\otimes\rho_{n}U_{\boldsymbol{F}}^{\dagger}}(\boldsymbol{v})W_{\mathbb{I}\otimes\dots\otimes\mathbb{I}\otimes E_{k_{n}}}(\boldsymbol{v})\nonumber \\
 & = & \sum_{\boldsymbol{v}\in\mathbb{Z}_{d}^{2n}}W_{\rho_{1}\otimes\dots\otimes\rho_{n}}(\boldsymbol{F}^{-1}\boldsymbol{v})W_{\mathbb{I}\otimes\dots\otimes\mathbb{I}\otimes E_{k_{n}}}(\boldsymbol{v}).\label{eq:nth_outcome_quant}
\end{eqnarray}
 Where we have recast the inner product into the discrete Wigner form
for convenience of comparison. We must now establish that the classical
circuit has the same distribution.

Classically, if the system is initially at point $\boldsymbol{v}$
on the discrete phase space then probability of getting outcome $k_{n}$
from the simulation circuit is given by: 
\begin{eqnarray*}
\text{Pr}_{\text{class}}(k_{n}|\boldsymbol{v}\text{ sampled initially}) & = & \text{Pr}_{\text{class}}(k_{n}|\boldsymbol{Fv}\text{ final location})\\
 & = & W_{\mathbb{I}\otimes\dots\otimes\mathbb{I}\otimes E_{k_{n}}}(\boldsymbol{Fv}).
\end{eqnarray*}
Which just says that the system is moved from point $\boldsymbol{v}$
to point $\boldsymbol{Fv}$ and the probability of outcome $k_{n}$
is the probability we see the system when we look at the region of
phase space measured by $\boldsymbol{E}_{k_{n}}$, which is $W_{\boldsymbol{E}_{k_{n}}}(\boldsymbol{Fv})$
by definition. The total probability of outcome $k_{n}$ is then:
\begin{eqnarray}
\text{Pr}_{\text{class}}(k_{n}) & = & \sum_{\boldsymbol{v}\in\mathbb{Z}_{d}^{2n}}\text{Pr}_{\text{class}}(\boldsymbol{k}|\boldsymbol{v}\text{ sampled init.})\Pr{\boldsymbol{v}\text{ sampled init.}}\nonumber \\
 & = & \sum_{\boldsymbol{v}\in\mathbb{Z}_{d}^{2n}}W_{E_{k_{1}}\otimes\dots\otimes E_{k_{n}}}(\boldsymbol{Fv})W_{\rho_{1}\otimes\dots\otimes\rho_{n}}(\boldsymbol{v})\nonumber \\
 & = & \sum_{\boldsymbol{v}\in\mathbb{Z}_{d}^{2n}}W_{E_{k_{1}}\otimes\dots\otimes E_{k_{n}}}(\boldsymbol{v})W_{\rho_{1}\otimes\dots\otimes\rho_{n}}(\boldsymbol{F}^{-1}\boldsymbol{v}).\label{eq:nth_outcome_classical}
\end{eqnarray}
Comparing Algorithm Class \prettyref{eq:nth_outcome_quant}, the distribution of measurement
outcomes on the last register for the quantum circuit, and Algorithm Class \prettyref{eq:nth_outcome_classical},
the simulated distribution of measurement outcomes on the last register,
we see they are the same. 

If the quantum algorithm is independent of the measurement outcomes
then simply applying the above argument to each register would suffice
to complete the proof. However, in general adaptive schemes are possible,
such the algorithm illustrated in \prettyref{fig:Example-Circuit-1}
where the final gate applied depends on the outcome of the measurement
on the third qudit. Using the assumption that the registers are measured
from last to first we can factor the distribution of outcome strings
as 
\[
\Pr{\boldsymbol{k}}=\Pr{k_{1}|k_{2}\dots k_{n}}\Pr{k_{2}|k_{3}\dots k_{n}}\dots\Pr{k_{n-1}|k_{n}}\Pr{k_{n}}.
\]
Since the simulation conditions on measurement outcome in exactly
the same way as the original quantum algorithm a simple inductive
argument shows that the distribution of outcomes must be the same
for the quantum algorithm and its classical simulator. \end{proof}
\begin{cor}
Quantum algorithms belonging to Algorithm Class \ref{alg:Simulable-Quantum-Circuit}
offer no super linear advantage over classical computation.\end{cor}
\begin{proof}
We have seen that if it is computationally efficient (linear in the
number of qudits) to sample from the classical distributions corresponding
to the input state and the measurements then such quantum circuits
are efficiently simulable. Since we have assumed separability of the
input and measurements and the discrete Wigner function factors this
efficient sampling is guaranteed.
\end{proof}
A couple of remarks are in order. We have restricted ourselves to
separable inputs and measurements, but this is not strictly necessary
for efficient simulation. Any positively represented preparation or
measurement can be accommodated provided it is possible to classically
efficiently sample from the corresponding distribution. Since it is
exponentially difficult to even write down general quantum states
this is rather strong restriction. 

Also, notice that our simulation protocol only samples from the output
distribution of a circuit whereas the Gottesman-Knill protocol gives
the full quantum state output in the case where the input is a pure
stabilizer state. It may appear that the present protocol is weaker
in this respect. However, the discrete Wigner function of pure stabilizer
states are uniformly valued lines on the phase space\cite{Gross2006Hudsons}
and these are fully specified by only two points. If we are promised
the input state to a Clifford circuit is a stabilizer state then we
can sample two distinct points from the corresponding distribution
and determine where the circuit maps them. These two output points
then suffice to fix the line corresponding to the output stabilizer
state. 

Finally, we note that, in the context of magic state distillation
for example, one may increase the size of the input register conditional
on measurement outcomes. This can be accounted for in the simulation
protocol above by simply increasing the size of the phase space accordingly
and sampling from the new additional positive Wigner functions.

\section{Magic State Distillation}

The main significance of the simulation result just established is
that for noisy preparation and measurement it is possible to extend
the efficient simulation of quantum circuits beyond the purview of
the stabilizer formalism. This result is of major theoretical importance,
but it also has practical significance. In particular, the simulation
scheme addresses the magic state model, which supplements error free
stabilizer resources with high fidelity additional gates produced
through the consumption of non-stabilizer ancillas. Recall that the
backbone of this process is a distillation protocol that uses stabilizer
resources applied to a large number of ancilla input states to produce
a few highly pure non-stabilizer states. An immediate corollary of
our simulation protocol is that states with positive discrete Wigner
representation are not useful for computational speedup in the magic
state model. Since this includes a large class of states outside the
stabilizer formalism this offers a resolution to the long standing
open problem of whether \emph{all} non-stabilizer states promote stabilizer
computation to universal quantum computation\cite{Bravyi_magic_state_paper}.

The class of algorithms in the previous section encompass a large
variety of magic state distillation protocols, but it is still conceptually
unclear why states with positive discrete Wigner representation are
not useful for magic state distillation. This is especially true since we typically think
of the outcome of a magic state distillation routine as a quantum
state, rather than a string of measurement outcomes as in the simulation
protocol. If we did keep track of the full input distribution then
we would be able to use it to reconstruct the quantum state output
of a distillation procedure; unfortunately this is impossible to do
efficiently, but a little thought shows it is not actually necessary.
We do not need to know the final quantum state, we only need to know
that it is not helpful for doing quantum computation. The simulation
protocol makes it clear that if the quantum state that is put into
the circuit has a Wigner function that is a genuine probability distribution
then the quantum state that is output (which we measure) must also
correspond to a genuine probability distribution. Inspired by this
observation, and in view of the great importance of magic state protocols,
we devote this section to a direct proof that negative discrete Wigner
representation of the ancilla resource states is necessary for such
states to be distillable (using stabilizer resources) to a non-stabilizer
state of arbitrary purity. 

The essential insight of the proof used
here is that negative discrete Wigner representation is a resource
that can not be created using stabilizer operations; if the input
states to a distillation protocol has no negativity in its discrete
Wigner representation then the output will not either. This resource
character is one of the major insights of the present work. Beyond
its conceptual value the alternative proof presented below also closes
several loopholes and alternative models not addressed by the simulation
protocol, which we discuss at the end of the section.

Conventional magic state protocols perform a Clifford unitary on the
input state $\rho_{\text{in}}^{\otimes n}\in L(\mathbb{C}_{d}\otimes\mathbb{C}_{d}^{n-1})$
and make a computational basis measurement on the final $n-1$ qudits,
post selecting on the $\ket 0$ outcome. This outputs the state, 
\[
\rho_{\text{out}}=\frac{\Tr_{\mathbb{C}_{d}^{n-1}}\left(\mathbb{I}\otimes\ketbra 00^{\otimes n}U\rho_{\text{in}}^{\otimes n}U^{\dagger}\mathbb{I}\otimes\ketbra 00^{\otimes n}\right)}{\text{Normalization}},
\]
where the partial trace knocks out the ancilla systems and the normalization
in the denominator just guarantees $\Tr(\rho_{\text{out}})=1$. We
examine significantly more general protocols: instead of requiring
$n$ copies of a qudit input state we allow an arbitrary positively
represented input state, in place of a Clifford unitary we allow any
completely positive map which preserves the set of positively represented
states and in place of the computational basis measurement we allow
any positively represented projective measurement. Since positive
representation is convex this suffices to eliminate classical randomness
as a potential loophole. It also precludes choices of entangled stabilizer
measurements as these are still positively represented. For convenience
of presentation we define, 
\begin{eqnarray*}
F(\rho) & = & \text{min}_{\boldsymbol{u}}\Tr(A_{\boldsymbol{u}}\rho),
\end{eqnarray*}
which is the maximally negative point of the quasi-probability representation
of $\rho$ over the phase space. If the input state to a distillation
routine is positively represented (ie. $F(\rho)\ge0$) then its output
is also positively represented:

\begin{thm} Let $\rho_{\text{in}}$ be a density operator on a $n$
qudit Hilbert space such that $F(\rho_{\text{in}})\ge0$. Let $\Lambda$
be a (completely positive) map on this space for which $F(\rho)\ge0\implies F(\Lambda(\rho))\ge0$.
Let $P$ be a positively represented projector on this space. If $\rho_{\text{out}}$
is produced by acting on $\rho_{\text{in}}$ with $\Lambda$ and post
selecting a measurement on the last $n-1$ qudits on outcome $P$
then $F(\rho_{\text{out}})\ge0$.\end{thm} \begin{proof} Since we
can use Heisenberg-Weyl operations to cycle between phase point operations,
without loss of generality 
\begin{eqnarray*}
F(\tilde{\rho}) & = & \text{min}_{\boldsymbol{u}}\Tr(A_{\boldsymbol{u}}\rho)\\
 & = & \frac{\Tr(A_{\boldsymbol{0}}\otimes\mathbb{I}^{\otimes n-1}\cdot\mathbb{I}\otimes P\cdot\Lambda(\rho_{\text{in}})\cdot\mathbb{I}\otimes P)}{\text{Normalization}}\\
 & = & \frac{\Tr(A_{\boldsymbol{0}}\otimes P\cdot\Lambda(\rho_{\text{in}}))}{\text{Normalization}}.
\end{eqnarray*}
 The denominator is always positive so 
\[
F(\tilde{\rho})\ge0\iff\Tr(A_{\boldsymbol{0}}\otimes P\cdot\Lambda(\rho_{\text{in}}))\ge0.
\]

By assumption $F(\rho_{\text{in}})$ is positive, and thus so is $F(\Lambda(\rho_{\text{in}}))$.
We write $P=\sum_{\boldsymbol{v}}z_{\boldsymbol{v}}A_{\boldsymbol{v}}$
where, since $P$ is positively represented, $z_{\boldsymbol{u}}\ge0$
by the discrete Hudson's theorem. This gives us that, 
\[
\Tr(A_{\boldsymbol{0}}\otimes P\cdot\Lambda(\rho_{\text{in}}))=\sum_{\boldsymbol{v}}z_{\boldsymbol{v}}\Tr(A_{\boldsymbol{0}}\otimes A_{\boldsymbol{v}}\Lambda(\rho_{\text{in}})).
\]
 The non-negativity of $F(\Lambda(\rho_{\text{in}}))$ implies, by
definition, that $\Tr(A_{\boldsymbol{0}}\otimes A_{\boldsymbol{v}}\Lambda(\rho_{\text{in}}))\ge0$
so it must be the case that $\Tr(A_{\boldsymbol{0}}\otimes P\cdot\Lambda(\rho_{\text{in}}))\ge0$
and this implies $F(\tilde{\rho})\ge0.$ \end{proof} 

This proof has a few technical merits over the simulation result of
the previous section. Namely, it does not require the input state
to the distillation protocol to be separable or otherwise efficiently
sampleable and the update dynamics are not limited to only Clifford
unitary maps. Indeed it's not even necessary to have a description
of the input state or the transformation, all that is required is
the promise that $\Lambda(\rho_{\text{in}})$ has positive discrete
Wigner representation. Since many input states and channels do not
admit any efficient representation this means that there are resources
that are provably useless for distillation even when an efficient
simulation of the corresponding distillation process would be impossible.

\section{The Geometry of Positively Represented States}

The Gottesman-Knill theorem already establishes that the action of
(qubit) Clifford unitary operations on stabilizer states is efficiently
classically simulatable. Since the only pure states with positive
discrete Wigner representation are stabilizer states it is natural
to wonder if every positively represented state is a mixture of stabilizer
states. As we have already alluded to, remarkably this is \emph{false}.
To establish this we will clarify the geometry of the region of state
space which has positive representation and show that it strictly
contains the set of mixtures of stabilizer states. Combined with the
results of the previous sections this establishes our simulation protocol
as an extension of Gottesman-Knill and proves the existence of bound
states for magic state distillation, states which are not convex combinations
of stabilizer states but which are nevertheless not distillable using
perfect Clifford operations.

The set of convex combinations of stabilizer states is a convex polytope
with the stabilizer states as vertices. Any polytope can be defined
either in terms of its vertices or as a list of half space inequalities
called facets. Intuitively, these correspond to the faces of $3$
dimensional polyhedrons. We show that in power of prime dimension
each of the $d^{2}$ phase space point operators define a facet of
the stabilizer polytope. These are only a proper subset of the faces
of the stabilizer polytope, implying the existence of states with
positive representation which are not convex combinations of stabilizer
states. See Figure \ref{fig:Cartoon-Depicting-Intersection} for a cartoon
capturing the intuition for this result.

\begin{figure}
\includegraphics[width=0.5\columnwidth]{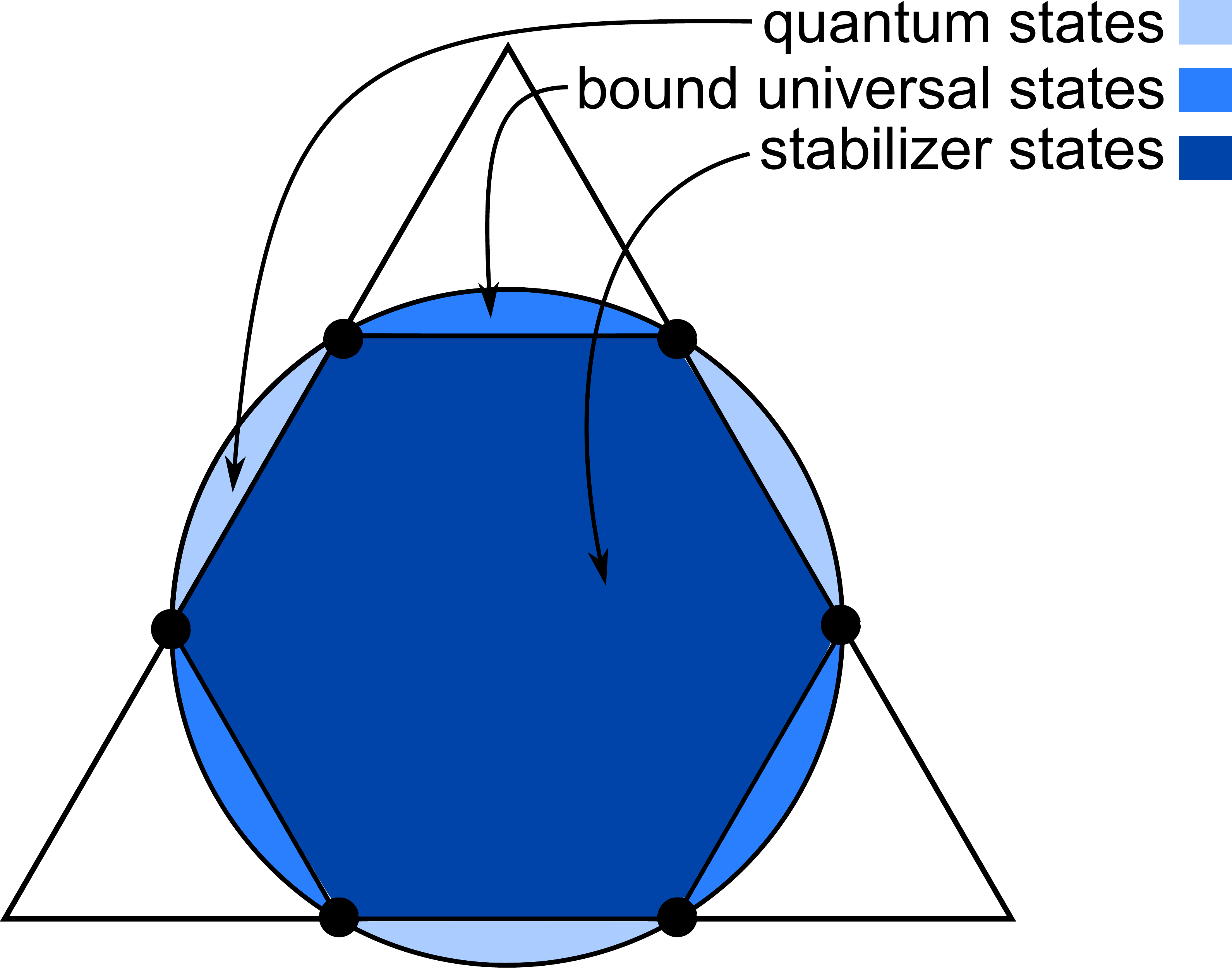}

\caption{A cartoon of the intersection of the discrete Wigner
probability simplex (the triangular region) with the quantum state
space (the circle). The simplex intersects the boundary at stabilizer
states (bold dots). The region of convex combinations of stabilizer
states is strictly contained within the set of
quantum states that also lie inside the simplex. The quantum states
outside the simplex are the bound states. Finally, the
quantum states with negative discrete Wigner representation are those lying outside the positive discrete Wigner simplex.  We show that the half space inequalities defining the facets of discrete
Wigner simplex also define the facets of the stabilizer polytope;
a fact reflected in this cartoon.\label{fig:Cartoon-Depicting-Intersection}}
\end{figure}

The stabilizer polytope may be thought of as a bounded convex polytope
living in $\mathbb{R}^{d^{2}-1}$, the space of $d$ dimensional mixed
quantum states. A minimal half space description for a polytope in
$\mathbb{R}^{D}$ is a finite set of bounding equalities called facets
$\{F_{i},f_{i}\}$ with $F_{i}\in\mathbb{R}^{D}$ and $f_{i}\in\mathbb{R}$.
$X\in\mathbb{R}^{D}$ is in the polytope if and only if $X\cdot F_{i}\le f_{i}\ \forall i$.
In the usual quantum state space the vectors $X$ of interest are
density matrices, the inner product is the trace inner product and
facets may be defined as $\{\hat{A}_{i},a_{i}\}$ where $\hat{A}_{i}$
are Hermitian matrices and 
\[
\rho\in\text{polytope}\iff\Tr(\rho\hat{A}_{i})\le a_{i}\ \forall i.
\]
 The objective is to show that $\{-A_{\boldsymbol{u}},0\}$ are facets
of the polytope defined by stabilizer state vertices.

It is possible to explicitly compute a facet description for a polytope
given the vertex description, but the complexity of this computation
scales polynomially in the number of vertices. Since the number of
stabilizer states grows super-exponentially with the number of qudits
\cite{Gross2006Hudsons} the conversion is generally impractical.
The analytic proof given here circumvents this issue. We also remark
that the work of Cormick \emph{et al} \cite{Cormick2006Classicality}
implies that the phase space point operators considered here are facets
for the case of \emph{prime} dimension. \begin{thm} The $d^{2}$
phase space point operators $\{A_{\boldsymbol{u}}\}$ with the inequalities
$\Tr(\rho A_{\boldsymbol{u}})>0$ define facets of the stabilizer
polytope. \end{thm} \begin{proof} To establish that a halfspace
inequality for a polytope in $\mathbb{R}^{D}$ is a facet there are
two requirements: every vertex must satisfy the inequality and there
must be a set of vertices saturating the inequality which span a space
of dimension $D$ \cite{Ziegler1995Lectures}.

The requirement that all vertices satisfy the half space inequality
is $\Tr(A_{\boldsymbol{u}}S)\ge0$ for every stabilizer state $S$,
and this the discrete Hudson's theorem.

We consider the stabilizer polytope as an object in $\mathbb{R}^{d^{2}-1}$
and look for a set of $d^{2}-1$ linearly independent vertices which
satisfy $\Tr(A_{\boldsymbol{u}}S)=0$. Since we are restricting to
power of prime dimension we may choose a complete set of mutually
unbiased bases of $d(d+1)$ states from the full set of stabilizer
states. Suppose more than $d+1$ states $V_{i}$ from this set satisfy
$\Tr(V_{i}A_{\boldsymbol{u}})>0.$ Then a counting arguments shows
that there must be two distinct states $V_{0},V_{1}$ belonging to
an orthonormal basis which satisfy this criterion. But then 
\begin{align*}
\Tr(V_{0}V_{1}) & =\frac{1}{d}\sum_{\boldsymbol{v}}\Tr(V_{0}A_{\boldsymbol{v}})\Tr(V_{1}A_{\boldsymbol{v}})\\
 & \ge\frac{1}{d}\Tr(V_{0}A_{\boldsymbol{u}})\Tr(V_{1}A_{\boldsymbol{u}})\neq0,
\end{align*}
 which contradicts the orthonormality. Thus at least $d(d+1)-(d+1)$
states in the mutually unbiased bases satisfy $\Tr(A_{\boldsymbol{u}}V_{i})=0$.
These are the required a set of $d^{2}-1$ linearly independent vertices.
\end{proof} The phase space point operators considered here give
only a proper subset of the defining halfspace inequalities for the
stabilizer polytope. This means that there are states that may not
be written as a convex combination of stabilizer states which nevertheless
satisfy $\Tr(A_{\boldsymbol{u}}\rho)\ge0$ for all phase space point
operators. That is, there are positive states which are not in the
convex hull of stabilizer states. These are bound states for magic
state distillation. An explicit example of such a state for the qutrit
is given in \cite{Gross2006Hudsons}.  These regions can be visualized by taking two and three dimensional slices of the qutrit state space.  Such slices are depicted in Figures \ref{fig:Projection-of-Postively} and \ref{fig:2D}.

\begin{figure}[ht]
\includegraphics[width=.9\columnwidth]{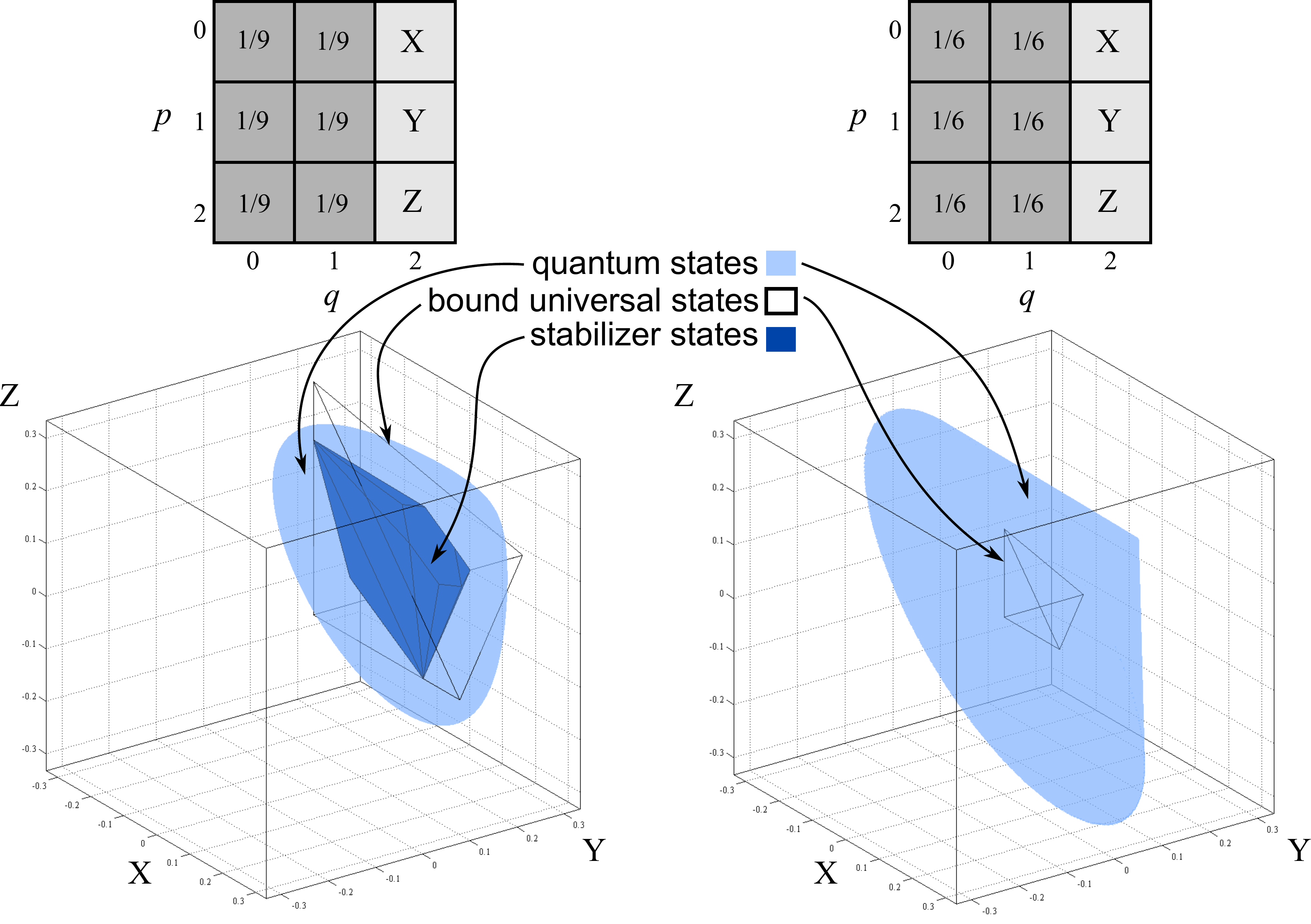}

\caption{Orthogonal 3-dimensional slices of qutrit state space.  Above each slice is the six values of the Wigner function which are fixed at a value of $1/9$ (left) and $1/6$ (right).  The remaining three values are allowed to vary and carve out regions depicted in the graphs.  In each case, due to symmetries, there are ${9 \choose 6} = 84$ such slices which are identical (up to a relabeling of the axes).  Note that the slice on the right does not cut through the stabilizer polytope but does contain a region of bound states.  Also note that this slice contains one of the $9$ states with a maximal negativity of $1/3$ while the slice on the left, and those equivalent up to permutations, are the only ones which feature the maximally mixed state $(X,Y,Z)=1/9$. See also Figure \ref{fig:2D} for 2-dimensional slice of the figure on the left.  \label{fig:Projection-of-Postively}}
\end{figure}

\begin{figure}[ht]
\includegraphics[width=0.75\columnwidth]{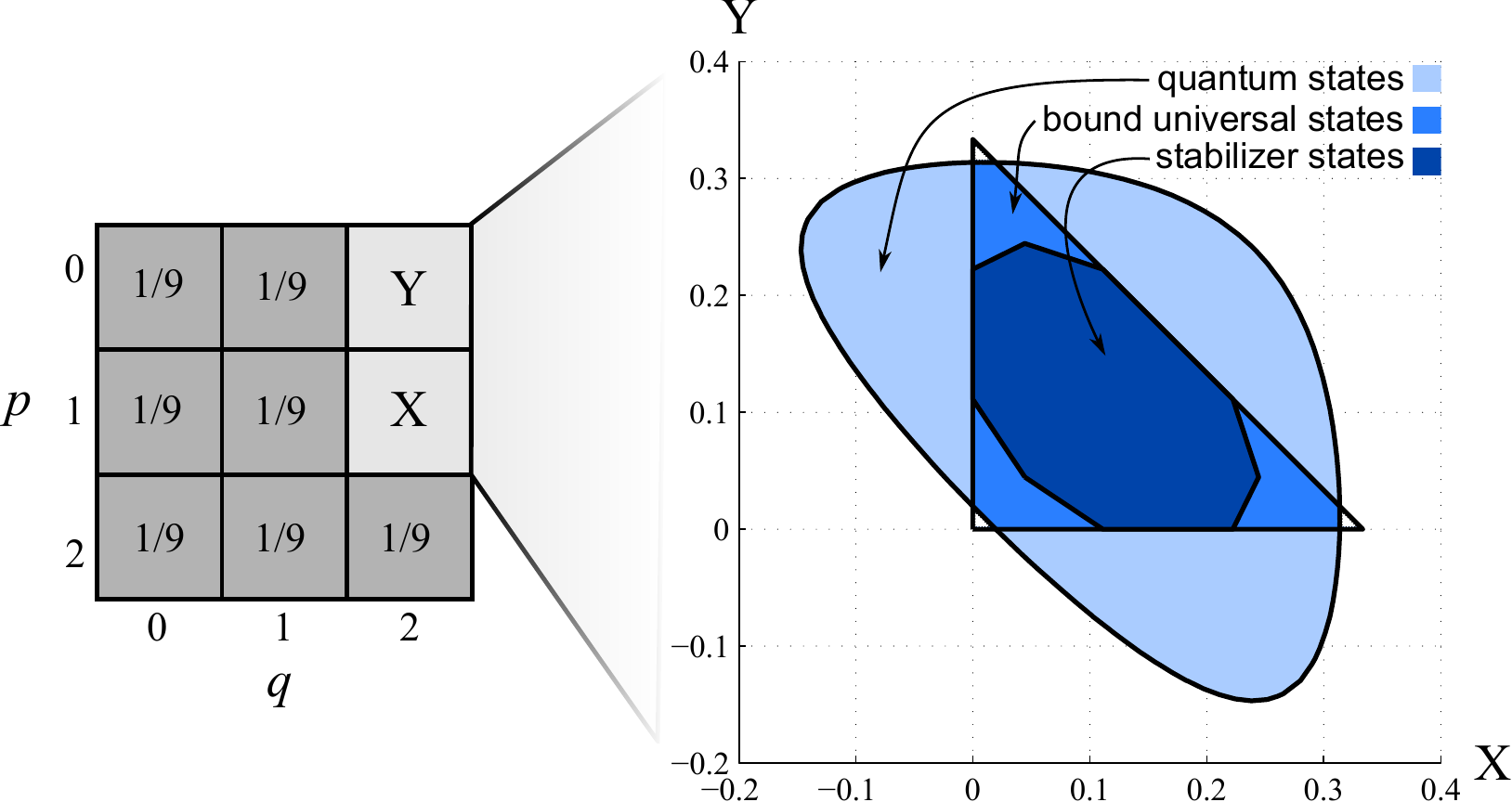}

\caption{Orthogonal 2-dimensional slice of qutrit state space.  On the left are the seven values of the Wigner function which are fixed at a value of $1/9$.  The remaining two are allowed to vary.  The maximally mixed state is the point $(X,Y) = (1,1)/9$.  The various regions carved out by varying these values are shown on the right.  There are ${9 \choose 7} = 36$ such slices which are identical (up to a relabeling of the axes).  These would be the only slices featuring the maximally mixed state.  Note the remarkable similarity the caricature in \ref{fig:Cartoon-Depicting-Intersection}.\label{fig:2D}}
\end{figure}

\section{Discussion and Conclusions}

We have shown that for systems of odd power of prime dimension a necessary
condition for computational speedup using Clifford unitaries is negativity
of the discrete Wigner representation of the inputs. This result is
immediately relevant in the context of magic state distillation,
where it shows that a necessary condition for distillability is negative
representation of the ancilla preparation. We have also shown that
the phase space point operators defining the discrete Wigner function
correspond to a privileged set of facets of the stabilizer polytope.
Together the two results imply the existence of non-stabilizer resources
which do not promote Clifford computation to universal quantum computation;
and in particular this establishes the existence of bound states for magic state distillation, or bound universal states.

We motivated the development of negative discrete Wigner representation
by analogy to entanglement theory, with Clifford operations playing
the role of local operations and classical communication and stabilizer states playing the role
of separable states. It is known that there are slightly entangled
mixed states that can not be consumed by distillation routines to
produce highly entangled states\cite{HorodeckMafiaBoundEnt}. The
non-stabilizer but positively represented quantum states are exactly analogous
to these bound entangled states. Similarly, it is known that for pure states large amounts
of entanglement are required for quantum computational speedup\cite{VidalEntang},
but for mixed states this is still an open question. However, for negative
discrete Wigner representation there is no relevant distinction between
mixed states and pure states. Moreover, although it is not yet known
whether every negatively represented state is distillable we conjecture
this to be the case.

The discrete Wigner function considered in this work is defined by analogy with the
more familiar continuous variable Wigner function. It is natural to
wonder if the results of this work extend to the infinite dimensional
case. In the continuous case a pure state has positive Wigner representation
if and only if it is a Gaussian state\cite{Hudson1974When} (like stabilizer
states in the discrete case) and a unitary evolution acts as a symplectic
flow on the phase space if and only if it corresponds to a quadratic
Hamiltonian \cite{Weedbrook2011Gaussian} (like Clifford gates
in the discrete case). A result of Brocker and Werner \cite{Brocker1995Mixed} shows that there
are mixed states that can not be represented as probabilistic combinations
of Gaussian states but that nevertheless have positive Wigner function.
The natural question is then: is it possible to efficiently simulate
quantum systems with quadratic Hamiltonian dynamics acting on non-Gaussian
mixed states if these states have positive representation?  The answer to this question, \emph{yes}, is obtained by giving
an explicit efficient simulation protocol for such systems \cite{thefuture}. 
This establishes that two of the main results of the present paper (efficient simulation and
non-stabilizer mixed states with positive discrete Wigner representation)
extend to the continuous case. It is interesting to ask if our third
result, that negative discrete Wigner representation is a necessary
resource for distillation, also has an analogue. That is, does the
continuous variable case admit something analogous to the magic state
model which allows noisy negatively represented states to be consumed
in order to promote linear optics to full universal quantum computational
power?

Of course there remains a final detail that we have not yet addressed.
The discrete Wigner function underpinning our analysis is only defined
for odd dimensional systems; is it possible to find a similar construction
for qubits? The discrete Wigner function used here has two crucial
properties: Clifford operations are stochastic transformations on
the underlying phase space and this phase space is separable. For
qubit systems there is no known analogue. Indeed it is easy to see
that a quasi-probability representation defined by any proper subset
of the facets of the qubit stabilizer polytope in a fashion analogous
to what has been done here will assign positive representation to
some subset of the magic states. The construction of the discrete
Wigner function, and the Clifford group, relies critically on the
mathematics of finite fields, and it is well known that fields of
characteristic 2 behave fundamentally differently than fields of any
other characteristic. This fact is reflected in the theory of error
correction where somewhat different protocols are required for dealing
with bits and qubits than are used for dits and qudits. In the case
of error correction, although qubits require more involved mathematics,
the conceptual underpinnings are the same irrespective of the underlying
dimensionality of the dits. It is not unreasonable to hope that a
similar result will hold for qubits and that an appropriately
modified mathematical strategy will preserve the conceptual
insights and related technical results obtained in the qudit case. If this turns out not to be the case
then understanding exactly why the model fails for qubits will undoubtedly
provide deep insights into the workings of quantum theory and quantum
computation.

The most interesting outstanding question raised by this work is whether
the ability to prepare any state with negative discrete Wigner representation
is sufficient to promote Clifford computation to universal quantum
computation. In \emph{prime} dimension the discrete Wigner construction
is the \emph{unique} choice of quasi-probability representation covariant
under the action of Clifford operations \cite{Gross2006Hudsons},
where the law of total probability is required to hold. On this basis
we conjecture that the condition here is sufficient. From the work
of Campbell, Anwar and Browne \cite{Campbell2012Magic} it is already known that access to any
non-stabilizer \emph{pure} state (or equivalently any negatively represented
pure state) suffices. If this conjecture is true then this implies an equivlance of two previously unrelated concepts of 
non-classicality, namely, quantum computational speedup and negative
quasi-probability representation.

\emph{Acknowledgements - } We thank Nathan Weibe, Daniel Gottesman
and Ben Reichardt for helpful comments.   After completion of this work we became aware of 
complementary work by Mari and Eisert \cite{mari_sim} which develops an alternative efficient simulation for
systems with positively represented Wigner function in odd and infinite dimenions. We acknowledge financial
support from CIFAR, the Government of Canada through NSERC, and
the Ontario Government through OGS and ERA. J.E. thanks the Perimeter Institute for Theoretical Physics where this work was brought to completion.

\bibliographystyle{plainnat}
\bibliography{magic_state,csferrie}

\end{document}